\pdfoutput=1
\documentclass[a4paper,UKenglish,thm-restate]{lipics-v2021}
\usepackage{algorithm,etoolbox,tikz}
\usepackage[noend]{algpseudocode}

\title{Faster Deterministic Modular Subset Sum}

\author{Krzysztof Potępa}{Jagiellonian University, Kraków}{krzysztof.potepa@student.uj.edu.pl}{}{}

\authorrunning{K. Potępa}

\Copyright{Krzysztof Potępa}



\acknowledgements{I would like to thank Lech Duraj, Krzysztof Pióro, and Adam Polak for their help with preparing the paper, and the anonymous reviewers for their useful suggestions.}

\nolinenumbers
\hideLIPIcs

\newtheorem{invariant}[theorem]{Invariant}

\newcommand\floor[1]{\left\lfloor#1\right\rfloor}
\newcommand\bigO{\mathcal{O}}

\newcommand\Init{\textnormal{\textsc{Init}}}
\newcommand\Set{\textnormal{\textsc{Set}}}
\newcommand\Shift{\textnormal{\textsc{Shift}}}
\newcommand\Diff{\textnormal{\textsc{Diff}}}
\newcommand\Update{\textnormal{\textsc{Update}}}
\newcommand\FindDifferences{\textnormal{\textsc{FindDifferences}}}
\newcommand\NewTag{\textnormal{\textsc{NewTag}}}
\newcommand\Find{\textnormal{\textsc{Find}}}
\newcommand\Union{\textnormal{\textsc{Union}}}
\newcommand\DeleteTag{\textnormal{\textsc{DeleteTag}}}
\DeclareMathOperator{\Level}{level}
\DeclareMathOperator{\Skew}{skew}
\DeclareMathOperator{\Left}{left}
\DeclareMathOperator{\Right}{right}
\DeclareMathOperator{\Parent}{parent}
\DeclareMathOperator{\Str}{str}
\DeclareMathOperator{\Bitrev}{bitrev}

\makeatletter
\patchcmd{\ALG@doentity}{\item[]\nointerlineskip}{}{}{}
\makeatother

\begin{document}

\maketitle

\begin{abstract}
    We consider the \emph{Modular Subset Sum} problem: given a multiset $X$ of integers from $\mathbb{Z}_m$ and a target integer $t$,
    decide if there exists a subset of $X$ with a sum equal to $t \pmod{m}$.
    Recent independent works by Cardinal and Iacono (SOSA'21), and Axiotis et al.\ (SOSA'21) provided simple and near-linear algorithms for this problem.
    Cardinal and Iacono gave a randomized algorithm that runs in $\bigO(m \log m)$ time,
    while Axiotis et al.\ gave a deterministic algorithm that runs in $\bigO(m \text{ polylog } m)$ time.
    Both results work by reduction to a text problem, which is solved using a dynamic strings data structure.

    In this work, we develop a simple data structure, designed specifically to handle the text problem that arises in the algorithms for Modular Subset Sum.
    Our data structure, which we call the \emph{shift-tree}, is a simple variant of a segment tree.
    We provide both a hashing-based and a deterministic variant of the \emph{shift-trees}.

    We then apply our data structure to the Modular Subset Sum problem and obtain two algorithms.
    The first algorithm is Monte-Carlo randomized and matches the $\bigO(m \log m)$ runtime of the Las-Vegas algorithm by Cardinal and Iacono.
    The second algorithm is fully deterministic and runs in $\bigO(m \log m \cdot \alpha(m))$ time, where $\alpha$ is the inverse Ackermann function.
\end{abstract}

\section{Introduction}
\label{sec:introduction}

The \emph{Subset Sum} is a fundamental problem in computer science.
It is defined as follows: given a multiset $X$ of $n$ positive integers and a target integer $t$,
decide if there exists a subset of $X$, such that the sum of its elements is exactly $t$.
The problem is known to be NP-complete \cite{Karp1972}, but only in a weak sense:
a classic dynamic programming approach of Bellman \cite{Bellman1957} solves it in pseudo-polynomial $\bigO(nt)$ time.
In recent years, there has been a lot of research towards improving the runtime \cite{Koiliaris2017,Koiliaris2019,Bringmann2017,Jin2019},
which culminated in near-linear $\widetilde{\bigO}(n+t)$ algorithms \cite{Bringmann2017,Jin2019}.%
\footnote{By writing $\widetilde{\bigO}(f(n))$, we mean $\bigO(f(n) \text{ polylog } f(n))$.}

In this work, we focus on the \emph{Modular Subset Sum} problem.
The \emph{Modular Subset Sum} is a natural variant of the Subset Sum problem, where all sums are taken modulo $m$, for some given modulus $m$.
We assume that the input multiset $X$ is provided in a compact form: as a list of $\bigO(m)$ distinct elements along with their multiplicities.
This assumption allows us to omit dependence on the number of elements $n$ in algorithm complexities.
Moreover, we focus on algorithms that return all possible subset sums, i.e. a set of all attainable values of $t$.

The dynamic programming of Bellman \cite{Bellman1957} can be easily adapted to solve the modular case in $\bigO(nm)$ time.
Let $S_i$ be the set of all attainable subset sums using only the first $i$ elements.
Bellman's algorithm iteratively computes the sets $S_1, ..., S_n$ using formula $S_i = S_{i-1} \cup (S_{i-1} + x_i)$,
where $x_i$ is the $i$-th input element and $S_{i-1} + x_i = \{a+x_i : a \in S_{i-1}\}$.
Most of the currently known improved algorithms simply simulate the consecutive iterations of Bellman's algorithm faster.
An early notable exception is the $\widetilde{\bigO}(m^{5/4})$ algorithm of Koiliaris and Xu \cite{Koiliaris2017},
which uses a divide-and-conquer approach based on results from number theory.

Abboud et al.\ \cite{Abboud2019} obtained a SETH-based conditional lower bound for the Subset Sum problem,
which in particular implies that the Modular Subset Sum cannot be solved in $\bigO(m^{1-\varepsilon})$ time for any $\varepsilon > 0$.
The first randomized algorithm that matched their lower-bound (up to subpolynomial factors) was introduced by Axiotis et al.\ in \cite{Axiotis2019}.
They achieved a running time of $\bigO(m \log^7 m)$ by simulating Bellman's dynamic programming faster using ideas from linear sketching.

Recently, simple and practical algorithms were provided independently in \cite{Cardinal2021,Axiotis2021}.
Both results work by reducing the problem of computing Bellman's iteration to a text problem, but use different data structures to solve it efficiently.
A Las-Vegas randomized $\bigO(m \log m)$ algorithm by Cardinal and Iacono \cite{Cardinal2021}
uses the dynamic strings data structure of Gawrychowski et al.\ \cite{Gawrychowski2018}.
The authors also introduced a simpler alternative, called \emph{Data Dependent Trees}, with logarithmic bounds per operation.
On the other hand, Axiotis et al.\ \cite{Axiotis2021} obtained a deterministic $\bigO(m \text{ polylog } m)$ algorithm
by employing a deterministic data structure of Mehlhorn et al.\ \cite{Mehlhorn1997} instead.
More precisely, their algorithm is output-sensitive and works in $\bigO(|X^*| \text{ polylog } |X^*|)$ time, where $X^*$ is the set of all attainable subset sums.
The authors provided also a very simple, randomized $\bigO(m \log^2 m)$ algorithm that uses only an elementary prefix sum structure.

A very recent result of Bringmann and Nakos \cite{Bringmann2021} provides near-linear algorithms
for computing the sumset $A_1 + ... + A_n$, for $A_1, ..., A_n \subseteq \mathbb{Z}_m$.
This problem generalizes the Modular Subset Sum: the set of all attainable subset sums can be expressed as a sumset $\{0, x_1\} + ... + \{0, x_n\}$.

\subsection{Our contributions}

In this work, we develop a simple tree-based data structure, designed specifically to handle the text problem that arises in the algorithms for Modular Subset Sum.
Our data structure, which we call a \emph{shift-tree}, maintains a string $s$ under the following operations:
\begin{romanenumerate}
    \item change a single character of $s$;
    \item cyclically shift $s$ by $k$ positions;
    \item given another string $t$ with its corresponding shift-tree, and an interval $[a;b]$, list all positions in $[a;b]$ where strings $s$ and $t$ differ.
\end{romanenumerate}
We provide two variants of the data structure: a hashing-based one,
and a deterministic one with slightly worse time complexity (by $\alpha(n)$, where $\alpha$ is the inverse Ackermann function).
By applying shift-trees to the Modular Subset Sum problem, we obtain the following algorithms:

\begin{theorem}[name=,restate=MainTheoremRandomized]
    There exists an algorithm that returns all attainable modular subset sums of a multiset of integers from $\mathbb{Z}_m$ with high probability,
    in time $\bigO(m \log m)$ and space $\bigO(m)$.
\end{theorem}

\begin{theorem}[name=,restate=MainTheoremDeterministic]
    There exists a deterministic algorithm that returns all attainable modular subset sums of a multiset of integers from $\mathbb{Z}_m$
    in time $\bigO(m \log m \cdot \alpha(m))$ and space $\bigO(m)$.
\end{theorem}

The first variant is Monte-Carlo randomized and matches the runtime of Las-Vegas algorithm by Cardinal and Iacono \cite{Cardinal2021}.
The second variant is fully deterministic and improves upon the result of Axiotis et al.\ \cite{Axiotis2021}.
Our algorithms are offline as they process the input elements in specific order to achieve their running times.

Although we provide a detailed analysis only for Monte-Carlo randomized and deterministic shift-trees,
it is also possible to obtain a Las-Vegas implementation of the data structure.
Such an implementation automatically leads to a Las-Vegas algorithm for Modular Subset Sum that truly matches the runtime obtained in \cite{Cardinal2021}.
We outline this approach in Remark \ref{remark:las-vegas}.

\paragraph*{Sketch of the shift-tree data structure}

We now explain the high-level idea behind our data structure.
The shift-tree is a perfect binary tree built upon some string $s$.
The leaves of the tree store the consecutive letters of string $s$.
Since the tree is perfect, the length of string $s$ is required to be a power of two.
The inner nodes correspond to substrings of $s$ formed from underlying leaves and store their hashes.
The hashes can be updated in logarithmic time after changing a single character of $s$.

Consider two shift-trees $T_1$ and $T_2$ built for strings $s_1$ and $s_2$ respectively, such that $|s_1| = |s_2| = m$.
We can find all positions where $s_1$ and $s_2$ differ by descending from the roots of both trees simultaneously.
We compare hashes in the roots and proceed recursively with the children if the hashes differ.
Assuming there is no hash collision, we end up in leaves corresponding to positions where $s_1$ and $s_2$ differ.
Such a procedure will take $\bigO(k \log m)$ time, where $k$ is the number of differences.

The tricky operation is the cyclic shift of the maintained string.
A naive approach would be to simply rebuild the whole tree in linear time.
We improve this by noticing that some parts of the tree can be reused.
Assume that the string has length $m$ and we want to shift the string by $2^j$.
Such operation is equivalent to shifting subtrees of size $2^j$ by $1$, what can be done by changing links to children on the appropriate level of the shift-tree.
After such modification, hashes on higher levels still need to be updated, but not the hashes in the moved subtrees.
This yields a total time of $\bigO(m / 2^j)$ for a cyclic shift by $2^j$, and it can be easily extended to shifts of form $k2^j$.
Even though it seems like a subtle improvement, it is enough to obtain a fast algorithm for Modular Subset Sum.

To make the shift-trees deterministic, we replace hashes with \emph{tags}.
Tags are identifiers associated with strings, but unlike hashes, they are not unique: one string can be represented by multiple tags.
Each time a node is updated it receives a new tag.
We propagate the information about tags that represent the same strings lazily while searching for differences.
More specifically, if the tags are not known to be equal, the search procedure always recurs.
If the recursion was unnecessary, we know about it upon return and we can memorize that the respective tags were equivalent.

\paragraph*{Sketch of the algorithm for Modular Subset Sum}

Our algorithm follows the ideas of \cite{Axiotis2021,Cardinal2021}.
We simulate Bellman's algorithm faster.
We iteratively compute the sets of new attainable subset sums $C_i$ after adding the $i$-th element.
More precisely, $C_i = S_i \setminus S_{i-1} = \left(S_{i-1} + x_i\right) \setminus S_{i-1}$.
The key idea is to notice that instead of computing $C_i$, we can compute the symmetric difference $D_i = \left(S_{i-1} + x_i\right) \triangle S_{i-1}$,
and then reduce it to $C_i$, because $|D_i| = 2|C_i|$.

Let $s_i \in \{0, 1\}^m$ be the characteristic vector of the set $S_i$, i.e. $s_i[j] = 1$ iff $j \in S_i$.
The problem of finding the set $D_{i+1}$ is then reduced to the problem of finding differences between the string $s_i$ and its cyclic shift.
We apply shift-trees to solve this problem efficiently.
The shift-tree requires the length of the string to be a power of two, so we assume that $m = 2^k$ for now
(we show how to get rid of this assumption in section \ref{sec:modular-subset-sum}).
Consider two shift-trees $T_1$ and $T_2$ built for string $s_i$ and its cyclic shift respectively.
We simulate the Bellman's algorithm step as follows:
\begin{romanenumerate}
    \item adjust the cyclic shift of $T_2$;
    \item find the set $D_i$ by comparing $T_1$ and $T_2$;
    \item update the trees with new attainable subset sums.
\end{romanenumerate}
The bottleneck of the algorithm are the adjustments of cyclic shift of $T_2$:
if elements are processed in arbitrary order, the total complexity of shift operations can be $\bigO(m^2)$.
In section \ref{sec:traversing-cyclic-shifts}, we show that if elements are processed in a bit-reversal order,
then the shift operations amortize to $\bigO(m \log m)$.

\subsection{Preliminaries}

We introduce the following notation for strings.
We use the same notation for other sequences.

\begin{definition}
    Given a string $s = c_0...c_{n-1}$, we refer to $c_i$ as $s[i]$ and to substring $c_i...c_j$ as $s[i:j]$.
\end{definition}

\begin{definition}
    Given a string $s$ and $k \in \mathbb{Z}$, we denote by $s^{+k}$ and $s^{-k}$ the \emph{cyclic shift} of $s$ by $k$ positions to the right and left respectively.
    In other words, for every $i \in \{0, ..., |s|-1\}$:
    \begin{gather*}
        s[i] = s^{+k}[(i+k) \bmod |s|] = s^{-k}[(i-k) \bmod |s|]
    \end{gather*}
\end{definition}

\section{Shift-trees}
\label{sec:shift-tree}

\subsection{Overview}

We introduce \emph{shift-trees}, a variant of the segment tree data structure.
A \emph{shift-tree} $T$ maintains a string $s$ of length $m = 2^n$ over an alphabet $\Sigma$ and supports the following operations:
\begin{itemize}
    \item $T.\Init(s)$: Initialize the data structure with string $s$.
    \item $T.\Set(i, x)$: Given an index $i \in \{0, ..., |s|-1\}$ and a letter $x \in \Sigma$, change $s[i]$ to $x$.
    \item $T.\Shift(k)$: Given an offset $k \in \mathbb{Z}$, replace $s$ with $s^{+k}$, i.e. cyclically shift the string $s$ by $k$ positions to the right.
    \item
        $T.\Diff(Q, a, b)$: Given another shift-tree $Q$ representing a string $q$ such that $|s| = |q|$, list all differences between $s[a:b]$ and $q[a:b]$,
        i.e. return the list $L$ of all integers $x$ such that $a \leq x \leq b$ and $s[x] \neq q[x]$.
\end{itemize}
In this section, we describe a hashing-based version of the data structure,
which uses $\bigO(m)$ memory and supports these operations in the following time complexities:
\begin{itemize}
    \item $\Init$: $\bigO(m)$;
    \item $\Set$: $\bigO(\log m)$;
    \item $\Shift(k)$: $\bigO(m / 2^j)$, where $j$ is the largest integer such that $2^j \mid k$;
    \item
        $\Diff$: $\bigO((d+1) \log m)$,
        where $d = |L|$ is the number of differences.%
        \footnote{By writing $d+1$ in the complexity, we mean that the runtime of $\Diff$ operation is $\bigO(\log m)$ if $d = 0$.}
\end{itemize}
In section \ref{sec:deterministic-shift-tree}, we present a variant that is fully deterministic,
but achieves a slightly worse time complexity (by $\alpha(m)$, where $\alpha$ is the inverse Ackermann function).

We require an integer alphabet $\Sigma$ of size $\bigO(\text{poly}(m))$.
We use a standard Rabin-Karp rolling hash function \cite{Karp1987}:
we choose a sufficiently big prime $p$ and an integer $r \in \mathbb{Z}_p$, where $r$ is chosen uniformly at random.
The hash of a string $s$ is defined as $h(s) = \sum_{i=0}^{|s|-1} s[i] \cdot r^i \bmod p$.
We assume that $\Sigma \subseteq \mathbb{Z}_p$, so $h(x) = x$ for $x \in \Sigma$.
If hashes of two strings of the same length are equal, then the strings are equal with high probability.
Moreover, given hashes of some strings $s_1$ and $s_2$, one can compute hash of their concatenation using the following identity:
$h(s_1s_2) = h_1 + h_2 \cdot r^{|s_1|}$.
To enable constant-time computation of this formula, we precompute powers of $r$ up to $r^m$.
We use these properties extensively in our data structure.

\subsection{Structure}

Let $s$ be the string maintained by the data structure and let $|s| = 2^n$.
The shift-tree is a perfect binary tree built upon the string $s$.
The leaves of the tree store the consecutive letters of string $s$, with the leftmost one corresponding to $s[0]$ and the rightmost one corresponding to $s[|s|-1]$.
The inner nodes correspond to substrings of $s$ formed from underlying leaves and store hashes to enable their fast comparison.
By $\Level(v)$ we denote the distance from the node $v$ to the root node.
The root node has level $0$ and the leaves have level $n$.
There are $2^k$ nodes on the $k$-th level.

We now provide a compact $\bigO(m)$ memory representation of the data structure that enables us to achieve the desired complexity of $\Shift$ operation.
The only data stored in memory is an array of hashes $H[1 : 2^{n+1}-1]$ and a single integer $\Delta \in \{0, ..., 2^n-1\}$.
The values in $H[2^n : 2^{n+1}-1]$ correspond to the leaves and are letters of the represented string.
The value of $\Delta$ defines a cyclic shift of leaf indices, i.e. the $j$-th leftmost leaf of the tree has index $(j-\Delta) \bmod 2^n + 2^n$.

The tree structure is defined implicitly based on the value of $\Delta$ as follows.
The nodes of the tree are numbered from $1$ to $2^{n+1}-1$.
The nodes on the $k$-th level are numbered from $2^k$ to $2^{k+1}-1$.
In particular, the root node has index $1$ and the leaves have indices from $2^n$ to $2^{n+1}-1$.
We first introduce the following auxiliary function:
\begin{gather*}
    \Skew(k) = \floor{\frac{\Delta}{2^{n-k}}} \bmod 2
\end{gather*}
Note that floor division by a power of two is equivalent to right bitwise shift, so value of $\Skew(k)$ is simply $(n-k)$-th least significant bit of $\Delta$.
Let $i$ be an inner node and let $k = \Level(i)$. We define the children of node $i$ as follows:
\begin{align*}
    \Left(i) &= (2i - \Skew(k+1)) \bmod 2^{k+1} + 2^{k+1} \\
    \Right(i) &= (2i + 1 - \Skew(k+1)) \bmod 2^{k+1} + 2^{k+1}
\end{align*}
We also define the parent of node $i \neq 1$ at level $k$.
\begin{gather*}
    \Parent(i) = \floor{\frac{(i + \Skew(k)) \bmod 2^k + 2^k}{2}}
\end{gather*}
The \emph{left}, \emph{right} and \emph{parent} functions can be implemented in constant time.
The following lemma and corollary summarize the properties of a tree structure defined as above.

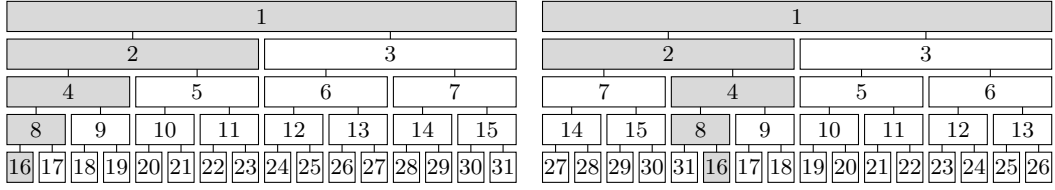
\begin{figure}[t]
    \begin{subfigure}[b]{0.495\textwidth}
        \centering
        \begin{tikzpicture}
    \draw[fill=gray!30] (1,0) rectangle node{\footnotesize 1} (7.7,0.4);
    \draw (2.655,-0.1) -- (2.655,0);
    \draw[fill=gray!30] (1,-0.5) rectangle node{\footnotesize 2} (4.31,-0.1);
    \draw (1.8075,-0.6) -- (1.8075,-0.5);
    \draw[fill=gray!30] (1,-1) rectangle node{\footnotesize 4} (2.615,-0.6);
    \draw (1.38375,-1.1) -- (1.38375,-1);
    \draw[fill=gray!30] (1,-1.5) rectangle node{\footnotesize 8} (1.7675,-1.1);
    \draw (1.171875,-1.6) -- (1.171875,-1.5);
    \draw[fill=gray!30] (1,-2) rectangle node{\footnotesize 16} (1.34375,-1.6);
    \draw (1.595625,-1.6) -- (1.595625,-1.5);
    \draw (1.42375,-2) rectangle node{\footnotesize 17} (1.7675,-1.6);
    \draw (2.23125,-1.1) -- (2.23125,-1);
    \draw (1.8475,-1.5) rectangle node{\footnotesize 9} (2.615,-1.1);
    \draw (2.019375,-1.6) -- (2.019375,-1.5);
    \draw (1.8475,-2) rectangle node{\footnotesize 18} (2.19125,-1.6);
    \draw (2.443125,-1.6) -- (2.443125,-1.5);
    \draw (2.27125,-2) rectangle node{\footnotesize 19} (2.615,-1.6);
    \draw (3.5025,-0.6) -- (3.5025,-0.5);
    \draw (2.695,-1) rectangle node{\footnotesize 5} (4.31,-0.6);
    \draw (3.07875,-1.1) -- (3.07875,-1);
    \draw (2.695,-1.5) rectangle node{\footnotesize 10} (3.4625,-1.1);
    \draw (2.866875,-1.6) -- (2.866875,-1.5);
    \draw (2.695,-2) rectangle node{\footnotesize 20} (3.03875,-1.6);
    \draw (3.290625,-1.6) -- (3.290625,-1.5);
    \draw (3.11875,-2) rectangle node{\footnotesize 21} (3.4625,-1.6);
    \draw (3.92625,-1.1) -- (3.92625,-1);
    \draw (3.5425,-1.5) rectangle node{\footnotesize 11} (4.31,-1.1);
    \draw (3.714375,-1.6) -- (3.714375,-1.5);
    \draw (3.5425,-2) rectangle node{\footnotesize 22} (3.88625,-1.6);
    \draw (4.138125,-1.6) -- (4.138125,-1.5);
    \draw (3.96625,-2) rectangle node{\footnotesize 23} (4.31,-1.6);
    \draw (6.045,-0.1) -- (6.045,0);
    \draw (4.39,-0.5) rectangle node{\footnotesize 3} (7.7,-0.1);
    \draw (5.1975,-0.6) -- (5.1975,-0.5);
    \draw (4.39,-1) rectangle node{\footnotesize 6} (6.005,-0.6);
    \draw (4.77375,-1.1) -- (4.77375,-1);
    \draw (4.39,-1.5) rectangle node{\footnotesize 12} (5.1575,-1.1);
    \draw (4.561875,-1.6) -- (4.561875,-1.5);
    \draw (4.39,-2) rectangle node{\footnotesize 24} (4.73375,-1.6);
    \draw (4.985625,-1.6) -- (4.985625,-1.5);
    \draw (4.81375,-2) rectangle node{\footnotesize 25} (5.1575,-1.6);
    \draw (5.62125,-1.1) -- (5.62125,-1);
    \draw (5.2375,-1.5) rectangle node{\footnotesize 13} (6.005,-1.1);
    \draw (5.409375,-1.6) -- (5.409375,-1.5);
    \draw (5.2375,-2) rectangle node{\footnotesize 26} (5.58125,-1.6);
    \draw (5.833125,-1.6) -- (5.833125,-1.5);
    \draw (5.66125,-2) rectangle node{\footnotesize 27} (6.005,-1.6);
    \draw (6.8925,-0.6) -- (6.8925,-0.5);
    \draw (6.085,-1) rectangle node{\footnotesize 7} (7.7,-0.6);
    \draw (6.46875,-1.1) -- (6.46875,-1);
    \draw (6.085,-1.5) rectangle node{\footnotesize 14} (6.8525,-1.1);
    \draw (6.256875,-1.6) -- (6.256875,-1.5);
    \draw (6.085,-2) rectangle node{\footnotesize 28} (6.42875,-1.6);
    \draw (6.680625,-1.6) -- (6.680625,-1.5);
    \draw (6.50875,-2) rectangle node{\footnotesize 29} (6.8525,-1.6);
    \draw (7.31625,-1.1) -- (7.31625,-1);
    \draw (6.9325,-1.5) rectangle node{\footnotesize 15} (7.7,-1.1);
    \draw (7.104375,-1.6) -- (7.104375,-1.5);
    \draw (6.9325,-2) rectangle node{\footnotesize 30} (7.27625,-1.6);
    \draw (7.528125,-1.6) -- (7.528125,-1.5);
    \draw (7.35625,-2) rectangle node{\footnotesize 31} (7.7,-1.6);
\end{tikzpicture}
        \caption{$\Delta = 0000_2 = 0$}
    \end{subfigure}
    \begin{subfigure}[b]{0.495\textwidth}
        \centering
        \begin{tikzpicture}
    \draw[fill=gray!30] (1,0) rectangle node{\footnotesize 1} (7.7,0.4);
    \draw (2.655,-0.1) -- (2.655,0);
    \draw[fill=gray!30] (1,-0.5) rectangle node{\footnotesize 2} (4.31,-0.1);
    \draw (1.8075,-0.6) -- (1.8075,-0.5);
    \draw (1,-1) rectangle node{\footnotesize 7} (2.615,-0.6);
    \draw (1.38375,-1.1) -- (1.38375,-1);
    \draw (1,-1.5) rectangle node{\footnotesize 14} (1.7675,-1.1);
    \draw (1.171875,-1.6) -- (1.171875,-1.5);
    \draw (1,-2) rectangle node{\footnotesize 27} (1.34375,-1.6);
    \draw (1.595625,-1.6) -- (1.595625,-1.5);
    \draw (1.42375,-2) rectangle node{\footnotesize 28} (1.7675,-1.6);
    \draw (2.23125,-1.1) -- (2.23125,-1);
    \draw (1.8475,-1.5) rectangle node{\footnotesize 15} (2.615,-1.1);
    \draw (2.019375,-1.6) -- (2.019375,-1.5);
    \draw (1.8475,-2) rectangle node{\footnotesize 29} (2.19125,-1.6);
    \draw (2.443125,-1.6) -- (2.443125,-1.5);
    \draw (2.27125,-2) rectangle node{\footnotesize 30} (2.615,-1.6);
    \draw (3.5025,-0.6) -- (3.5025,-0.5);
    \draw[fill=gray!30] (2.695,-1) rectangle node{\footnotesize 4} (4.31,-0.6);
    \draw (3.07875,-1.1) -- (3.07875,-1);
    \draw[fill=gray!30] (2.695,-1.5) rectangle node{\footnotesize 8} (3.4625,-1.1);
    \draw (2.866875,-1.6) -- (2.866875,-1.5);
    \draw (2.695,-2) rectangle node{\footnotesize 31} (3.03875,-1.6);
    \draw (3.290625,-1.6) -- (3.290625,-1.5);
    \draw[fill=gray!30] (3.11875,-2) rectangle node{\footnotesize 16} (3.4625,-1.6);
    \draw (3.92625,-1.1) -- (3.92625,-1);
    \draw (3.5425,-1.5) rectangle node{\footnotesize 9} (4.31,-1.1);
    \draw (3.714375,-1.6) -- (3.714375,-1.5);
    \draw (3.5425,-2) rectangle node{\footnotesize 17} (3.88625,-1.6);
    \draw (4.138125,-1.6) -- (4.138125,-1.5);
    \draw (3.96625,-2) rectangle node{\footnotesize 18} (4.31,-1.6);
    \draw (6.045,-0.1) -- (6.045,0);
    \draw (4.39,-0.5) rectangle node{\footnotesize 3} (7.7,-0.1);
    \draw (5.1975,-0.6) -- (5.1975,-0.5);
    \draw (4.39,-1) rectangle node{\footnotesize 5} (6.005,-0.6);
    \draw (4.77375,-1.1) -- (4.77375,-1);
    \draw (4.39,-1.5) rectangle node{\footnotesize 10} (5.1575,-1.1);
    \draw (4.561875,-1.6) -- (4.561875,-1.5);
    \draw (4.39,-2) rectangle node{\footnotesize 19} (4.73375,-1.6);
    \draw (4.985625,-1.6) -- (4.985625,-1.5);
    \draw (4.81375,-2) rectangle node{\footnotesize 20} (5.1575,-1.6);
    \draw (5.62125,-1.1) -- (5.62125,-1);
    \draw (5.2375,-1.5) rectangle node{\footnotesize 11} (6.005,-1.1);
    \draw (5.409375,-1.6) -- (5.409375,-1.5);
    \draw (5.2375,-2) rectangle node{\footnotesize 21} (5.58125,-1.6);
    \draw (5.833125,-1.6) -- (5.833125,-1.5);
    \draw (5.66125,-2) rectangle node{\footnotesize 22} (6.005,-1.6);
    \draw (6.8925,-0.6) -- (6.8925,-0.5);
    \draw (6.085,-1) rectangle node{\footnotesize 6} (7.7,-0.6);
    \draw (6.46875,-1.1) -- (6.46875,-1);
    \draw (6.085,-1.5) rectangle node{\footnotesize 12} (6.8525,-1.1);
    \draw (6.256875,-1.6) -- (6.256875,-1.5);
    \draw (6.085,-2) rectangle node{\footnotesize 23} (6.42875,-1.6);
    \draw (6.680625,-1.6) -- (6.680625,-1.5);
    \draw (6.50875,-2) rectangle node{\footnotesize 24} (6.8525,-1.6);
    \draw (7.31625,-1.1) -- (7.31625,-1);
    \draw (6.9325,-1.5) rectangle node{\footnotesize 13} (7.7,-1.1);
    \draw (7.104375,-1.6) -- (7.104375,-1.5);
    \draw (6.9325,-2) rectangle node{\footnotesize 25} (7.27625,-1.6);
    \draw (7.528125,-1.6) -- (7.528125,-1.5);
    \draw (7.35625,-2) rectangle node{\footnotesize 26} (7.7,-1.6);
\end{tikzpicture}
        \caption{$\Delta = 0101_2 = 5$}
    \end{subfigure}
    \caption{Shift-tree node numbering for different values of $\Delta$.}
\end{figure}

\begin{lemma}
    \label{lem:tree-struct}
    The functions left, right and parent define a perfect binary tree, such that the indices of nodes on the $k$-th level,
    when ordered from left to right, form a sequence $(2^k, ..., 2^{k+1}-1)^{+\floor{\Delta / 2^{n-k}}}$.
\end{lemma}
\begin{proof}
    We prove the lemma by induction on the level.
    The condition is satisfied for the $0$-th level, which contains only the root node with index $1$.
    Assume now that the condition is satisfied for the $k$-th level, i.e.
    the nodes on the $k$-th level form a sequence $(2^k, ..., 2^{k+1}-1)^{+\delta}$,
    where $\delta = \floor{\Delta / 2^{n-k}}$.
    By substituting indices of children in this sequence, we obtain the following sequence for the $(k+1)$-st level:
    \begin{gather*}
        (\Left(2^k), \Right(2^k), ..., \Left(2^{k+1}-1), \Right(2^{k+1}-1))^{+2\delta}
    \end{gather*}
    The shift is now $2\delta$, because each element has been replaced by two elements.
    We want to prove that this is exactly the sequence $(2^{k+1}, ..., 2^{k+2}-1)^{+\floor{\Delta / 2^{n-k-1}}}$.
    The shift $\floor{\Delta / 2^{n-k-1}}$ can be rewritten as $2\delta + x$, where $x = \Skew(k+1) = \floor{\Delta / 2^{n-k-1}} \bmod 2$.
    We now simplify the equation:
    \begin{align*}
        (\Left(2^k), \Right(2^k), ..., \Left(2^{k+1}-1), \Right(2^{k+1}-1))^{+2\delta} &\stackrel{?}{=} (2^{k+1}, ..., 2^{k+2}-1)^{+2\delta + x} \\
        (\Left(2^k), \Right(2^k), ..., \Left(2^{k+1}-1), \Right(2^{k+1}-1)) &\stackrel{?}{=} (2^{k+1}, ..., 2^{k+2}-1)^{+x} \\
        (\Left(2^k) - 2^{k+1}, \Right(2^k) - 2^{k+1}, ..., \Right(2^{k+1}-1) - 2^{k+1}) &\stackrel{?}{=} (0, ..., 2^{k+1}-1)^{+x}
    \end{align*}
    After substituting the values of \emph{left} and \emph{right}, and simplifying, we obtain the following:
    \begin{gather*}
        ((0-x) \bmod 2^{k+1}, (1-x) \bmod 2^{k+1}, ..., (2^{k+1}-1-x) \bmod 2^{k+1}) \stackrel{?}{=} (0, ..., 2^{k+1}-1)^{+x}
    \end{gather*}
    The obtained equation trivially satisfies the definition of cyclic shift by $x$, and all transformations were equivalent.
    This completes the induction.

    We complete the proof by showing that \emph{parent} function is well-defined.
    Consider an inner node on the $k$-th level with index $2^k+i$.
    It is enough to show that it is parent of its children.
    After substituting and simplifying the formulas, we get the desired result:
    \begin{align*}
        \Parent(\Left(2^k + i)) &= \floor{\frac{2i + 2^{k+1}}{2}} = 2^k + i \\
        \Parent(\Right(2^k + i)) &= \floor{\frac{2i + 1 + 2^{k+1}}{2}} = 2^k + i
    \end{align*}
\end{proof}

\begin{corollary}
    \label{cor:tree-struct}
    The functions left, right and parent define a perfect binary tree such that:
    \begin{alphaenumerate}
        \item\label{cor:tree-struct:a} nodes on the $k$-th level have indices from $2^k$ to $2^{k+1}-1$, for each valid $k$;
        \item\label{cor:tree-struct:b} the indices of leaves, when ordered from left to right, form a sequence $(2^n, ..., 2^{n+1}-1)^{+\Delta}$;
        \item\label{cor:tree-struct:c} the structure of subtrees rooted at the $k$-th level depends only on $\Delta \bmod 2^{n-k}$.
    \end{alphaenumerate}
\end{corollary}
\begin{proof}
    The first two properties follow instantly from the lemma \ref{lem:tree-struct}.
    For the property (\ref{cor:tree-struct:c}), notice that the links on these levels depend only on the values $\Skew(k+1), ..., \Skew(n)$.
    These values are exactly the $n-k$ least significant bits of $\Delta$.
\end{proof}

\subsection{Invariant}

Let $s$ be the string maintained by the data structure and let $|s| = 2^n$.
We define the string associated with a node $i$ recursively as follows:
\begin{gather*}
    \Str(i) =
        \begin{cases}
            s^{-\Delta}[i - 2^n] & \text{for } i \geq 2^n \text{ (i.e. $i$ is leaf node)} \\
            \Str(\Left(i)) \Str(\Right(i)) & \text{for } i < 2^n \text{ (i.e. $i$ is inner node)}
        \end{cases}
\end{gather*}
By corollary \ref{cor:tree-struct}\ref{cor:tree-struct:b}, the $k$-th letter of string $s$ is associated with the $k$-th leftmost leaf.
It follows that $s = \Str(1)$, i.e. string associated with the root node is $s$.
We maintain the following invariant:

\begin{invariant}
    For each node $i$ the following holds: $H[i] = h(\Str(i))$.
\end{invariant}

The invariant ensures that each node stores a hash of its associated string.
We use this property to implement the $\Diff$ operation in required time complexity.

\subsection{Operations}

Let $s$ be the string maintained by the data structure and let $|s| = m = 2^n$.
We first define an $\Update(i)$ primitive that is used by all operations that modify the data structure.
The $\Update$ procedure simply recalculates the hash of an inner node $i$ based on hashes of its children.
This can be done in constant time using basic modular arithmetic, if appropriate powers of $r$ are precomputed.

\begin{algorithm}[H]
    \caption{The $\Update$ procedure}
    \begin{algorithmic}[1]
        \Function{Update}{$i$}
            \State $H[i] \gets (H[\Left(i)] + H[\Right(i)] \cdot r^{|\Str(\Left(i))|}) \bmod p$
        \EndFunction
    \end{algorithmic}
\end{algorithm}

\subparagraph*{Init.}

We initialize $\Delta$ with $0$ and leaves with the letters of the input string $s$.
Specifically, we set $H[2^n + i] = s[i]$ for each $i \in \{0, ..., 2^n-1\}$, because letter $s[i]$ is associated with the node $2^n + i$.
We then compute all the hashes by calling an $\Update$ on the remaining nodes, beginning at the bottom of the tree.
Overall, the $\Init$ operation updates $\bigO(m)$ nodes and runs in $\bigO(m)$ time.

\subparagraph*{Set.}

Assume that we change $s[i]$ to $x$.
Let $j = (i - \Delta) \bmod 2^n + 2^n$.
Notice that, $\Str(j) = s^{-\Delta}[j - 2^n] = s[(j + \Delta) \bmod 2^n] = s[i]$.
In order to fix the invariant, we set $H[j] = x$ and update hashes of all the ancestors of $j$.
The tree has $\bigO(\log m)$ levels, so the total runtime of $\Set$ operation is $\bigO(\log m)$.

\noindent
\begin{minipage}[b]{0.49\textwidth}
    \begin{algorithm}[H]
        \caption{$\Init$ operation}
        \begin{algorithmic}[1]
            \Function{Init}{$s$}
                \State $\Delta \gets 0$
                \For{$i \gets 0, ..., 2^n-1$}
                    \State $H[2^n + i] \gets s[i]$
                \EndFor
                \For{$i \gets 2^n-1, ..., 1$}
                    \State $\Update(i)$
                \EndFor
            \EndFunction
        \end{algorithmic}
    \end{algorithm}
\end{minipage}
\hfill
\begin{minipage}[b]{0.49\textwidth}
    \begin{algorithm}[H]
        \caption{$\Set$ operation}
        \begin{algorithmic}[1]
            \Function{Set}{$i, x$}
                \State $j \gets (i - \Delta) \bmod 2^n + 2^n$
                \State $H[j] \gets x$
                \While{$j \neq 1$}
                    \State $j \gets \Parent(j)$
                    \State $\Update(j)$
                \EndWhile
            \EndFunction
        \end{algorithmic}
    \end{algorithm}
\end{minipage}

\subparagraph*{Shift.}

Assume that we apply a right cyclic shift by $k$ positions to $s$.
Let $j$ be the largest integer such that $2^j \mid k$.
In order to fix the invariant, we first set $\Delta$ to $(\Delta + k) \bmod 2^n$.
Notice that the invariant is now satisfied for leaves.
Moreover, by corollary \ref{cor:tree-struct}\ref{cor:tree-struct:c} the structure of subtrees rooted at level $n-j$ didn't change,
so the invariant is also satisfied for levels $n-j, ..., n$.
It remains to update hashes on the remaining $n-j$ levels by calling the $\Update$ procedure on their nodes.
Overall, $\bigO(2^{n-j})$ nodes are updated and the $\Shift$ operation runs in $\bigO(m / 2^j)$ time.%
\footnote{
    Provided algorithm requires computation of $j = \max \{ d \in \mathbb{N} : 2^d \mid k \}$.
    In practice, it is sufficient to compute $2^j$ instead of computing $j$ directly.
    This can be done in constant time using the following bit-hack: \lstinline|k & ~(k-1)|.
}

\begin{algorithm}[H]
    \caption{$\Shift$ operation}
    \begin{algorithmic}[1]
        \Function{Shift}{$k$}
            \If{$k \bmod 2^n = 0$}
                \State \Return
            \EndIf
            \State $j \gets \max \{ d \in \mathbb{N} : 2^d \mid k \}$
            \State $\Delta \gets (\Delta + k) \bmod 2^n$
            \For{$i \gets 2^{n-j}-1, ..., 1$}
                \State $\Update(i)$
            \EndFor
        \EndFunction
    \end{algorithmic}
\end{algorithm}

\subparagraph*{Diff.}

Assume we look for differences between the strings maintained by the trees $T$ and $Q$ in interval $[a;b]$.
We provide a recursive procedure $\FindDifferences(T, Q, a, b, i, j, x, y)$
that returns required set of differences between substrings associated with node $i$ of the tree $T$ and node $j$ of the tree $Q$.
The procedure additionally tracks an interval $[x;y]$ that is associated with both nodes, i.e. $T.\Str(i) = T.s[x:y]$ and $Q.\Str(j) = Q.s[x:y]$.

The $\FindDifferences$ procedure works as follows.
If $[a;b] \cap [x;y] = \emptyset$ then procedure returns empty set instantly, because we only look for differences in the interval $[a;b]$.
If hashes of nodes $i$ and $j$ are equal then the substrings are equal w.h.p., so the procedure returns no differences as well.
Otherwise, there is at least one difference between the strings associated with nodes $i$ and $j$.
If the nodes are leaves, then we report the difference.
If the nodes are inner nodes, the procedure is invoked recursively on left and right children.

The $\Diff$ operation simply calls $\FindDifferences$ on roots of the trees $T$ and $Q$.
The procedure will return all the required differences as long as there is no hash collision.

\begin{algorithm}[H]
    \caption{The $\Diff$ operation}
    \begin{algorithmic}[1]
        \Function{T.Diff}{$Q, a, b$}
            \State \Return $\FindDifferences(T, Q, a, b, 1, 1, 0, 2^n-1)$
        \EndFunction
        \State
        \Function{FindDifferences}{$T, Q, a, b, i, j, x, y$}
            \If{$[a;b] \cap [x;y] = \emptyset$ or $T.H[i] = Q.H[j]$}
                \State \Return $\emptyset$
            \EndIf
            \If{$x = y$} \Comment{The nodes $i$ and $j$ are leaves if they represent an unit interval}
                \State \Return $\{x\}$
            \EndIf
            \State $z \gets \frac{x+y+1}{2}$ \Comment{The left nodes represent $[x:z-1]$ and the right nodes represent $[z:y]$}
            \State $A \gets \FindDifferences(T, Q, a, b, T.\Left(i), Q.\Left(j), x, z-1)$
            \State $B \gets \FindDifferences(T, Q, a, b, T.\Right(i), Q.\Right(j), z, y)$
            \State \Return $A \cup B$
        \EndFunction
    \end{algorithmic}
\end{algorithm}

We now argue the complexity of $\Diff$ operation.
Let $d$ be the number of differences that were found.
Let $s_k$ be the number of $\FindDifferences$ calls for which the processed nodes were on the $k$-th level and the procedure recurred.
If the procedure recurred, there existed at least one difference between the strings associated with the nodes.
We can divide such calls into two categories:
\begin{alphaenumerate}
    \item\label{diff-calls-a} there is a difference in $[x;y] \cap [a;b]$ that should be reported;
    \item\label{diff-calls-b} there is a difference in $[x;y] \setminus [a;b]$ that should be ignored and $[x;y] \cap [a;b] \neq \emptyset$.
\end{alphaenumerate}
The number of calls that belong to the category (\ref{diff-calls-a}) is bounded by $d$, i.e. number of reported differences.
There are at most $2$ calls that belong to the category (\ref{diff-calls-b}), because the interval $[x;y]$ must contain $a$ or $b$.
It follows that $s_k \leq d+2$ and $s_0 + ... + s_{n-1} \leq (d+2) \cdot n$.
We can charge the calls that didn't recur to their parents, so the total running time of $\Diff$ operation is $\bigO((d+1) \log m)$.

In order to complete the analysis, we bound the probability that $\Diff$ operation fails to report all differences.
Such situation may occur only if nodes processed by $\FindDifferences$ have different associated strings, but equal hashes.
Let $s_1 \neq s_2$ be strings of length $k$.
\begin{gather*}
    \Pr\left[h(s_1) = h(s_2)\right] = \Pr\left[\sum_{i=0}^{k-1} (s_1[i]-s_2[i]) r^i \equiv 0 \pmod{p}\right] \leq \frac{k}{p}
\end{gather*}
The inequality holds, because the sum on the left side is a non-zero polynomial of degree at most $k$, evaluated in randomly chosen point $r$.
By application of union bound, we obtain that the probability of failure is at most $m \log m / p$.
Assuming operations on hashes of size $\bigO(\log m)$ are taking constant time,
we can choose $p = \Theta(\text{poly}(m))$ and obtain high probability of success.

\begin{remark}
    \label{remark:las-vegas}
    It is also possible to achieve a Las-Vegas implementation of shift-trees.
    The key observation is that the hash function doesn't need to be associative, i.e. one can hash ``subtrees'' instead of substrings.
    This allows us to replace the hash function with any injective mapping $h(x, y)$ from pairs of hashes into new hashes.
    Only the $\Update$ procedure needs to be adapted to compute the hash $H[i]$ of $i$-th node as $h(H[\Left(i)], H[\Right(i)])$.

    The missing piece is how to implement the mapping $h(x, y)$.
    This can be done by simply generating it on demand, and storing the mapping in a hashtable.
    Some garbage collection mechanism (such as reference counting) is required to maintain linear memory usage.
    This yields a Las-Vegas implementation of shift-trees, with the same expected runtime bounds.

    By replacing hashtable with a BST, one can obtain a deterministic implementation of the data structure.
    Such modification introduces logarithmic runtime overhead.
    In the next section, we improve upon this by allowing amortization.
\end{remark}

\section{Deterministic shift-trees}
\label{sec:deterministic-shift-tree}

\subsection{Overview}

We now provide a deterministic variant of the data structure introduced in previous section.
Let $m = 2^n$.
Assume that we maintain several shift-trees $T_1, ..., T_r$ associated
with strings $s_1, ..., s_r \in \Sigma^m$ of the same length, allowing comparisons between those strings.
Let $K = |s_1| + ... + |s_r| = mr$, and let $T_i$ be one of the trees.
Then, we can do the shift-tree operations on $T_i$ with the following amortized time complexities:
\begin{itemize}
    \item $\Init$: $\bigO(m \cdot \alpha(K))$;
    \item $\Set$: $\bigO(\log m \cdot \alpha(K))$;
    \item $\Shift(k)$: $\bigO(m / 2^j \cdot \alpha(K))$, where $j$ is the largest integer such that $2^j \mid k$;
    \item $\Diff$: $\bigO((d+1) \log m \cdot \alpha(K))$, where $d$ is the number of differences.
\end{itemize}
Moreover, the data structures use $\bigO(K)$ memory overall.
The only constraint that we put on alphabet $\Sigma$ is the support for equality tests.
This contrasts with the randomized variant, where an integer alphabet of polynomial size is required.

\subsection{Tags}

We replace hashing with the concept of \emph{tags}.
Tags are simply identifiers associated with strings.
Unlike hashes, they are not unique: one string can be represented by multiple tags.
Each inner node of the shift-tree stores a tag instead of a hash.
A new tag is created every time a node is updated.

We denote the string associated with a tag $t$ by $\Str(t)$.
The strings associated with tags are not stored in memory.
Instead, we maintain an equivalence relation $\mathcal{R}$ over the set of tags used in all the shift-trees.
The relation $\mathcal{R}$ satisfies the following property:
\begin{invariant}\label{inv:tags}
    For each tag $a$ and tag $b$ such that $a \equiv_\mathcal{R} b$, the strings $\Str(a)$ and $\Str(b)$ are equal.
\end{invariant}
Note that the inverse doesn't need to hold.
In other words, the relation $\mathcal{R}$ partially captures the equality relation between strings associated with tags.
The relation is refined during operations on the shift-trees using the following operations:
\begin{itemize}
    \item $\NewTag()$: Create a new tag $x$ with its own singleton equivalence class, and return $x$.
    \item $\Find(x)$: Given a tag $x$, return identifier of its equivalence class.
    \item
        $\Union(x, y)$: Given tags $x$ and $y$, union their equivalence classes
        (i.e. insert $x \equiv_\mathcal{R} y$ and close the relation transitively).
    \item $\DeleteTag(x)$: Given a tag $x$, remove it from its equivalence class and free its memory.
\end{itemize}
In order to support these operations efficiently, we represent the equivalence classes of $\mathcal{R}$ using a union-find data structure.
A simple extension of standard disjoint-set forest implementation with support for element removal has been proposed by Kaplan et al.\ in \cite{Kaplan2002}.

\begin{theorem}[\cite{Kaplan2002}]
    There exists a data structure that maintains an equivalence relation $\mathcal{R}$ using linear memory under operations
    $\Find$, $\Union$ and $\DeleteTag$ in amortized $\bigO(\alpha(n))$ time, and $\NewTag$ in $\bigO(1)$ time,
    where $n$ is the number of maintained elements.
\end{theorem}

Their approach is based on lazy deletions:
elements to be deleted are marked and the union-find trees are rebuilt if the fraction of marked elements is greater than half.
More involved approaches with constant time deletions have been known in literature \cite{Alstrup2005,Ben2011},
but such improvement doesn't change the amortized time complexity of our data structure.

The idea to use union-find data structure for detecting mismatches has been already proposed by Gawrychowski et al. in \cite{Gawrychowski2016}.

\begin{remark*}
    In general, the $\DeleteTag$ operation is not only useful for space optimization.
    If unused tags are not removed, the complexity of operations is dependent on the total number of $\NewTag$ calls, which can be large.
    This is not an issue if only $\bigO(\text{poly}(K))$ tags are created in total, because $\bigO(\alpha(\text{poly}(K))) = \bigO(\alpha(K))$.
\end{remark*}

\subsection{Operations}

We now adapt the $\Update$ procedure and $\Diff$ operation to work with tags.
The $\Init$, $\Set$ and $\Shift$ operations use the $\Update$ primitive and don't require changes.

\subparagraph*{Update.}

We simply create a new tag for the updated node.
If the node already contains a tag (i.e. the $\Update$ is called after initialization), we delete it in order to maintain linear memory usage.
The newly created tag is in a singleton equivalence class of $\mathcal{R}$, so it trivially satisfies the invariant.

\begin{algorithm}[H]
    \caption{The $\Update$ procedure}
    \begin{algorithmic}[1]
        \Function{Update}{$i$}
            \If{$H[i] \neq \text{null}$}
                \State $\DeleteTag(H[i])$
            \EndIf
            \State $H[i] \gets \NewTag()$
        \EndFunction
    \end{algorithmic}
\end{algorithm}

\subparagraph*{Diff.}

We adapt the $\FindDifferences$ procedure as follows.
Assume that we are looking for differences in interval $[a;b]$.
Let $t_1$ and $t_2$ be tags in compared tree nodes, and $[x;y]$ the interval associated with these nodes.
If $t_1 \equiv_\mathcal{R} t_2$ then the compared substrings are equal, so we exit instantly.
Otherwise, we search for differences recursively in the left and right subtrees.
If no differences are found and $[x;y] \subseteq [a;b]$, we know that $\Str(t_1) = \Str(t_2)$,
so we can safely add $t_1 \equiv_\mathcal{R} t_2$ to relation $\mathcal{R}$ via $\Union(t_1, t_2)$.

\begin{algorithm}[H]
    \caption{The $\FindDifferences$ procedure}
    \begin{algorithmic}[1]
        \Function{FindDifferences}{$T, Q, a, b, i, j, x, y$}
            \If{$[a;b] \cap [x;y] = \emptyset$} \Comment{Check if we are outside of the search interval}
                \State \Return $\emptyset$
            \EndIf
            \If{$x = y$} \Comment{The nodes $i$ and $j$ are leaves if they represent an unit interval}
                \If{$T.H[i] \neq Q.H[j]$} \Comment{The leaves contain letters -- compare them directly}
                    \State \Return $\{x\}$ \label{found-diff}
                \Else
                    \State \Return $\emptyset$
                \EndIf
            \EndIf
            \If{$\Find(T.H[i]) = \Find(Q.H[j])$} \Comment{The inner nodes contain tags -- use $\mathcal{R}$.}
                \State \Return $\emptyset$ \Comment{$T.H[i] \equiv_\mathcal{R} Q.H[j]$ holds, so the strings are equal}
            \EndIf
            \State $z \gets \frac{x+y+1}{2}$ \Comment{The left nodes represent $[x;z-1]$ and the right nodes represent $[z;y]$}
            \State $A \gets \FindDifferences(T, Q, a, b, T.\Left(i), Q.\Left(j), x, z-1)$
            \State $B \gets \FindDifferences(T, Q, a, b, T.\Right(i), Q.\Right(j), z, y)$
            \If{$A \cup B = \emptyset$ and $[x;y] \subseteq [a;b]$} \label{union-cond}
                \State $\Union(T.H[i], Q.H[j])$
            \EndIf
            \State \Return $A \cup B$
        \EndFunction
    \end{algorithmic}
\end{algorithm}

We now briefly address the correctness of the $\FindDifferences$ procedure. Observe that if the condition $[a;b] \cap [x;y] \neq \emptyset$ holds then:
\begin{romanenumerate}
    \item the invariant \ref{inv:tags} guarantees that the procedure will recur if substrings are not equal;
    \item the procedure returns a difference for leaves iff their corresponding characters differ.
\end{romanenumerate}
It follows by an easy induction on the level that the procedure returns all positions in $[a;b]$ where the strings differ and nothing more.
Moreover, the procedure doesn't break the invariant when modifying the relation $\mathcal{R}$:
if the condition in line \ref{union-cond} is true, then there are no differences between compared substrings.

\subsection{Running time}

Let $K$ be the sum of lengths of strings maintained by all shift-trees and let $m = 2^n$ be the length of each string.
We first note that the number of elements maintained by relation $\mathcal{R}$ never exceeds the total number of nodes in shift-trees, which is $\bigO(K)$,
so any operation on $\mathcal{R}$ works in amortized $\bigO(\alpha(K))$ time.

We now argue the amortized running time of the $\Update$ and $\Diff$ operations.
Let the actual cost of the $i$-th operation be a number $c_i$ of operations on the relation $\mathcal{R}$.
Let $q_i$ be the number of equivalence classes of relation $\mathcal{R}$ after $i$ operations.
We define potential $\Phi_i$ to be $9q_i$.
Clearly, $\Phi_i \geq \Phi_0 = 0$.
The amortized cost of the $i$-th operation is $\widehat{c_i} = c_i + \Phi_i - \Phi_{i-1}$.

The actual cost of an $\Update$ operation is $c_i \leq 2$.
The operation creates at most one new equivalence class, thus the amortized cost is $\widehat{c_i} = c_i + \Phi_i - \Phi_{i-1} \leq 2 + 9 \cdot 1 = 11$.
It follows that the amortized time complexity of an $\Update$ operation is $\bigO(\alpha(K))$.

To estimate the amortized cost of $\Diff$ operation, we consider the $\FindDifferences$ calls that recurred.
We say that the call is \emph{wasted} if the compared strings are equal, but comparison of tags reported that they are not.
Otherwise, if the compared strings are not equal and the procedure recurred, we say that the call is \emph{required}.
Let $w$ be the number of wasted calls and $r$ be the number of required calls.
We can charge the calls that didn't recur to their parents, so the total number of $\FindDifferences$ calls is bounded by $3(r+w)$.
Each call does at most three operations on relation $\mathcal{R}$, so the cost of $\Diff$ operation is $c_i \leq 9(r+w)$.

Let $d$ be the number of differences that have been found.
We can bound the number of required calls $r$ by $\bigO((d+1) \log m)$, the same way as in hashing-based shift-trees.
We now focus our attention on the wasted calls.
Assume that we are looking for differences in interval $[a;b]$.
Let $w'$ be the number of wasted calls such that the interval associated with compared nodes is contained within $[a;b]$.
Notice that upon return, each such call unions two different equivalence classes, thus decreasing the potential by $9$.
It follows that $\Diff$ operation reduces the potential in total by $9w'$.
On the other hand, the intervals associated with the remaining wasted calls must contain $a$ or $b$, so there are at most $2$ such calls for each tree level.
It follows that the number of all wasted calls is bounded by $w' + 2\log m$.
The amortized cost of $\Diff$ operation is then:
\begin{gather*}
    \widehat{c_i} \leq 9r + 9w - 9w' \leq 9r + 18\log m = \bigO((d+1) \log m)
\end{gather*}
and the amortized time complexity is $\bigO((d+1) \log m \cdot \alpha(K))$.

We complete analysis by providing amortized running time of $\Init$, $\Set$ and $\Shift$ operations.
The $\Init$ operation calls $\Update$ operation $\bigO(m)$ times, so its amortized running time is $\bigO(m \cdot \alpha(K))$.
By the same argument we obtain the required amortized complexities for $\Set$ and $\Shift$ operation.

\section{Traversing all cyclic shifts}
\label{sec:traversing-cyclic-shifts}

In this section, we consider a problem of going over all the cyclic shifts of the shift-tree efficiently, in some order.
This means that we want to consider all shifts $s^{+\sigma(0)}, ..., s^{+\sigma(|s|-1)}$, for $\sigma$ being some permutation of $\{0, ..., |s|-1\}$.
This requires invoking $\Shift(\sigma(i) - \sigma(i-1))$ for $i = 1,2,...,|s|-1$, assuming the shift-tree initially represents $s^{+\sigma(0)}$.
We claim that there exists a permutation $\sigma$ such that the total complexity of these operations amortizes to $\bigO(m \log m)$ time.
For simplicity, we consider the hashing-based shift-trees; the complexity for deterministic variant is just multiplied by $\alpha(K)$.
This technique is crucial for our Modular Subset Sum algorithm and might be used for other problems, where the order of operations doesn't matter.

The time complexity of a single $\Shift$ operation depends heavily on the value of shift.
Recall that the running time of $\Shift(k)$ is $\bigO(m / 2^j)$, where $j$ is the largest integer such that $2^j \mid k$, and $m = 2^n$ is the size of the shift-tree.
For example, a shift by $1$ requires rebuilding the entire tree, while shift by $m/2$ takes constant time.
It means that the complexity of going over all the cyclic shifts heavily depends on the permutation $\sigma$.

It turns out that a good permutation is a bit-reversal permutation, which we define as follows.
We denote the bit-reverse of $n$-bit number $j$ by $\Bitrev_n(j)$,
i.e. if $j = \sum_{i=0}^{n-1} c_i2^i$ then $\Bitrev_n(j) = \sum_{i=0}^{n-1} c_i2^{n-i-1}$.
We say that $\sigma$ is a \emph{bit-reversal permutation} if $\sigma(i) = \Bitrev_n(i)$.

\begin{lemma}
    \label{lem:shift-amortization}
    Let $\sigma$ be a bit-reversal permutation of length $m = 2^n$ and $T$ a shift-tree of length $\bigO(m)$.
    The sequence of operations $T.\Shift(\sigma_i - \sigma_{i-1})$ for $i = 1, ..., m-1$ takes total time $\bigO(m \log m)$.
\end{lemma}
\begin{proof}
    Let $\delta_i = \sigma_i - \sigma_{i-1}$ and consider a single $\Shift(\delta_i)$ operation.
    Let $j$ be the largest integer such that $2^j \mid \delta_i$.
    The complexity of this operation is then $\bigO(m / 2^j) = \bigO(2^{n-j})$.
    If $2^j \mid \delta_i$ and $2^{j+1} \nmid \delta_i$ then $j$ is the least significant bit that is different between $\sigma_i$ and $\sigma_{i-1}$.
    Since $\sigma_i$ is a bit-reverse of $i$, it means that $n-j-1$ is the most significant bit that is different between $i$ and $i+1$.
    Such situation happens only if $i+1$ is of form $k \cdot 2^{n-j-1}$, where $k$ is odd.
    There are $2^j$ such numbers in range $[1; 2^n-1]$.
    It follows that shifts for a given value of $j$ take overall $\bigO(2^{n-j} \cdot 2^j) = \bigO(2^n)$ time.
    There are $\bigO(n)$ possible values of $j$, so the whole sequence of shifts takes $\bigO(2^n \cdot n) = \bigO(m \log m)$ time.
\end{proof}

\section{Modular Subset Sum}
\label{sec:modular-subset-sum}

In this section, we provide an algorithm for the Modular Subset Sum problem that uses the shift-tree data structure.
Let $X = \{x_1, ..., x_n\}$ be a multiset of integers from $\mathbb{Z}_m$.
Our algorithm computes the set $X^* \subseteq \mathbb{Z}_m$ such that $k \in X^*$
if and only if there exists a subset of $X$ that sums to the value $k$ modulo $m$.
We assume that the input multiset $X$ is provided in a compact form: as a list of $\bigO(m)$ distinct elements along with their multiplicities.

The algorithm is based on the so-called Bellman's iteration.
Consider sets $S_0, ..., S_n \subseteq \mathbb{Z}_m$ such that $S_0 = \{0\}$ and $S_i = S_{i-1} \cup \left(S_{i-1} + x_i\right)$,
where $S_{i-1} + x_i = \{a+x_i : a \in S_{i-1}\}$.
It is easy to see that the set $S_i$ is a set of all attainable subset sums of $\{x_1, ..., x_i\}$ and the final result is $X^* = S_n$.
In order to compute the set $S_i$ from $S_{i-1}$, it is sufficient to find the set $C_i = S_i \setminus S_{i-1} = \left(S_{i-1} + x_i\right) \setminus S_{i-1}$.
If one can compute the set $C_i$ in time $\bigO(|C_i| \cdot f(m))$, then the set $X^*$ can be computed in total time $\bigO(m \cdot f(m))$.

The key idea of \cite{Axiotis2019} is to notice that instead of computing $C_i = \left(S_{i-1} + x_i\right) \setminus S_{i-1}$,
we can compute the symmetric difference $D_i = \left(S_{i-1} + x_i\right) \triangle S_{i-1}$.
Computing the set $D_i$ doesn't break the time complexity, because its size is only two times larger than $C_i$.
We can now interpret problem of finding the set $D_i$ as a text problem \cite{Axiotis2021,Cardinal2021}.
Let $s_i \in \{0, 1\}^m$ be the characteristic vector of the set $S_i$, i.e. $s_i[j] = 1$ iff $j \in S_i$.
The problem of finding the set $D_{i+1}$ is then reduced to problem of finding differences between the strings $s_i$ and $s_i^{+x_{i+1}}$.
The set $D_{i+1}$ is exactly the set of indices, where these strings differ.

We now describe our algorithm.
Let $S$ be the set of all subset sums attainable using elements processed so far, and let $s$ be its characteristic vector.
Initially, the set $S$ contains only $0$.
We maintain the characteristic vector $s$ and its cyclic shift using two shift-trees, $T_1$ and $T_2$ respectively.
The length of a string maintained by a shift-tree is required to be a power of two, but that may not be the case with the string $s$.
We address this issue in the following way.
Let $L$ be the smallest power of two such that $L \geq 2m$.
We consider the following auxiliary strings:
\begin{align*}
    s' &= s0^{L-m} \\
    s'' &= s0^{L-2m}s
\end{align*}
The strings $s'$ and $s''$ have length $L$, which is a power of two.
Moreover, the string $s''$ has the following property: the string $s^{+d}$ is a prefix of $(s'')^{+d}$, for any $d \in \{0, ..., m\}$.
This property allows us to find differences between $s$ and $s^{+d}$ by comparing a prefix of $s'$ with a prefix of $(s'')^{+d}$.

The tree $T_1$ maintains the string $s'$ and the tree $T_2$ maintains a cyclic shift of the string $s''$.
Specifically, we traverse all cyclic shifts of $T_2$ in bit-reversal order, as explained in the previous section.
Assume that the current shift of $T_2$ is $x$, i.e. the string maintained by $T_2$ is $(s'')^{+x}$.
Let $\mu$ be the multiplicity of $x$ in the input multiset $X$.
If $\mu = 0$, then $x \notin X$ and the algorithm proceeds to the next shift.
Otherwise, we simulate $\mu$ Bellman's iterations for the element $x$ as follows.
The $x$ is contained in the multiset $X$, so $x \in \{0, ..., m-1\}$.
It implies that we can find the set of differences $D$ between $s$ and $s^{+x}$ by comparing a prefix of $s'$ with a prefix of $(s'')^{+x}$.
This is done using $\Diff$ operation on $T_1$ and $T_2$.
We then update the set of attainable subset sums $S$ and both shift-trees appropriately.
If no differences were found, we skip the rest of iterations for element $x$.

Every element of $X$ corresponds to some cyclic shift, so all elements will be processed if all the cyclic shifts are considered.
The set $S$ is then the set of all attainable subset sums for $X$.
We provide the pseudocode as Algorithm \ref{alg:modular-subset-sum}.
In the pseudocode, we denote the multiplicity of element $x$ in the set $X$ by $\mu_X(x)$.

\begin{algorithm}[H]
    \caption{The \textsc{ModularSubsetSum} algorithm}
    \label{alg:modular-subset-sum}
    \begin{algorithmic}[1]
        \Function{ModularSubsetSum}{$X, m$}
            \State $S = \{0\}$
            \State $k \gets \min \{ d \in \mathbb{N} : 2^d \geq 2m \}$
            \State $L \gets 2^k$
            \State $s_0 \gets 10^{m-1}$ \Comment{The characteristic vector of $S_0 = \{0\}$}
            \State Initialize shift-tree $T_1$ with $s_00^{L-m}$
            \State Initialize shift-tree $T_2$ with $s_00^{L-2m}s_0$
            \For{$i \gets 1, ..., L-1$}
                \State $x \gets \Bitrev_k(i)$
                \State $T_2.\Shift(x - \Bitrev_k(i-1))$ \Comment{Now $T_2$ represents the string $(s'')^{+x}$}
                \For{$j \gets 1, ..., \mu_X(x)$}
                    \State $D \gets T_1.\Diff(T_2, 0, m-1)$ \Comment{$D$ is the set of all differences between $s$ and $s^{+x}$}
                    \If{$D = \emptyset$}\label{skip-cond} \Comment{Skip the rest of iterations for $x$ if no differences were found}
                        \State \textbf{break}
                    \EndIf
                    \For{$d \in D \setminus S$}
                        \State $S \gets S \cup \{d\}$
                        \State $T_1.\Set(d, 1)$
                        \State $T_2.\Set((d+x) \bmod L, 1)$
                        \State $T_2.\Set((d+x-m) \bmod L, 1)$
                    \EndFor
                \EndFor
            \EndFor
            \State \Return $S$
        \EndFunction
    \end{algorithmic}
\end{algorithm}

We now analyse the running time of the algorithm.
We assume the hashing-based shift-trees are used; the complexity for deterministic variant is just multiplied by $\alpha(m)$.
The initialization of shift-trees takes $\bigO(m)$ time.
The value of $\Bitrev_k(i)$ can be computed naively bit by bit in $\bigO(k) = \bigO(\log m)$ time, which in total takes $\bigO(m \log m)$ time.
Moreover, the total time of all $\Shift$ operations amortizes to $\bigO(m \log m)$ time due to lemma \ref{lem:shift-amortization}.
We now focus on the total running time of inner loops.

Consider a single Bellman's iteration.
The complexity of a $\Diff$ operation is $\bigO((|D|+1) \log m)$.
The algorithm adds $|D \setminus S| = |D|/2$ new elements to the set $S$ and updates the shift-trees.
Each tree update takes $\bigO(\log m)$ time.
In total, a single Bellman's iteration takes $\bigO((|D|+1) \log m)$ time.

The sum of sizes of all the sets of differences is at most $2m$.
It means that if there are $k$ Bellman's iterations in total, then their total running time is $\bigO((m+k) \log m)$.
The condition in the line \ref{skip-cond} ensures that the total number of executed Bellman's iterations is $\bigO(m)$ by skipping the iterations if the set is empty.
It follows that all the iterations take $\bigO(m \log m)$ time in total.

We arrive at the total time complexity of $\bigO(m \log m)$.
By replacing hashing-based shift-trees with their deterministic variant, we obtain a deterministic algorithm with running time of $\bigO(m \log m \cdot \alpha(m))$.
We recall the theorems that summarize these results:

\MainTheoremRandomized*
\MainTheoremDeterministic*

\bibliography{modular-subset-sum}

\end{document}